\documentclass[times]{TRR}

\usepackage{moreverb,url}
\usepackage{hyperref}
\usepackage{natbib}
\usepackage{cleveref}
\usepackage{graphicx}
\usepackage{amsthm}
\newtheorem{proposition}{Proposition}
\newcommand\BibTeX{{\rmfamily B\kern-.05em \textsc{i\kern-.025em b}\kern-.08em
T\kern-.1667em\lower.7ex\hbox{E}\kern-.125emX}}

\def\volumeyear{2025}

\begin{document}

\runninghead{Yue et al.}

\title{Signal Timing Optimization for Mixed Connected Automated Traffic Based on A Markov Delay Approximation}

\author{Ximin Yue\affilnum{1}, Yunlong Zhang\affilnum{1}, Zihao Li\affilnum{1}, and Yang Zhou\affilnum{1}}

\affiliation{\affilnum{1}Zachry Department of Civil \& Environmental Engineering, Texas A\&M University, College Station, TX 77843, USA.}

\corrauth{Zihao Li, scottlzh@tamu.edu}

\begin{abstract}
Connected Automated Vehicles (CAVs) offer unparalleled opportunities to revolutionize existing transportation systems. In the near future, CAVs and human-driven vehicles (HDVs) are expected to coexist, forming a mixed traffic system. Although several prototype traffic signal systems leveraging CAVs have been developed, a simple yet realistic approximation of mixed traffic delay and optimal signal timing at intersections remains elusive. This paper presents an analytical approximation for delay and optimal cycle length at an isolated intersection of a mixed traffic using a stochastic framework that combines Markov chain analysis, a car following model, and queuing theory. Given the intricate nature of mixed traffic delay, the proposed framework systematically incorporates the impacts of multiple factors, such as the distinct arrival and departure behaviors and headway characteristics of CAVs and HDVs, through mathematical derivations to ensure both realism and analytical tractability. Subsequently, closed-form expressions for intersection delay and optimal cycle length are derived. Numerical experiments are then conducted to validate the model and provide insights into the dynamics of mixed traffic delays at signalized intersections.

\vspace{6pt}
\end{abstract}

\maketitle

\section{INTRODUCTION}
Delay experienced by HDVs at signalized intersections are often exacerbated by slower reaction times, lack of coordination with signal control systems, and the need to maintain a large headway \citep{wang2019review}. The advancement of CAV technologies, particularly in vehicular automation and vehicle-to-everything (V2X) communication,offers promising potential to reduce traffic delays \citep{naghsh2018delay, khattak2023impact, ma2022lyapunov, allison2024performance, yue2024hybrid}. Recent research has increasingly focused on optimizing CAV trajectories to enhance traffic flow, improve intersection throughput, and increase energy efficiency \citep{li2018piecewise, ma2017parsimonious, zhou2017parsimonious, guo2019joint}. With the continued rise in CAV market penetration rate (MPR), these capabilities are expected to yield substantial reductions in delays at signalized intersections \citep{ghiasi2017mixed}.

Intersection delay, used as the primary indicator for assessing the level of service (LOS) at signalized intersections, has been extensively studied in the context of CAVs. Prevailing analyses often heavily rely on simulation-based methodologies \citep{baz2020intersection}. However, the simulation-based approach fails to provide a comprehensive understanding of the detailed impacts of CAV design on signalized interaction delay, such as those involving CAV control parameters. Still, it also suffers from limitations due to its computational intensity. On the other hand, analytical delay models serve as essential tools for evaluating control performance and adjusting signal plans by transportation agencies. Most existing analytical delay models for signalized intersections are based on HDV-dominated traffic and can be traced back to as early as 1952. When \cite{wardrop1952road} estimated the average HDV delay by assuming a uniform traffic arrival distribution and a linear relationship between traffic flow and density. Subsequently, \cite{webster1966traffic} developed Webster's delay model based on deterministic queuing theory to empirically calculate the average delay and derive the corresponding optimal cycle length by differentiating the equation for overall delay at the intersection with respect to the cycle length. Later, a more general model based on empirical adjustment, the Highway Capacity Manual (HCM) delay model \citep{manual2000highway}, was developed. One of its key elements is the HCM delay model, which is used to estimate delays experienced by vehicles at different types of facilities, such as intersections, roundabouts, or freeway segments. 

Leveraging insights from previous studies, several models have been proposed to better approximate intersection delay and optimize signal timing. \cite{abhigna2022delay} developed a multiple linear regression model to estimate delays at uncontrolled intersections in mixed traffic by comparing various traditional models such as HCM, Drew's model, and Harder's model \citep{drew1968traffic, harders1968capacity}. \cite{su2024cycle} derived an optimal cycle length model by incorporating time-variant effects into Webster’s delay formula using an adjustment factor method. \cite{ali2021adaptive} presented an adaptive traffic signal control approach using fuzzy logic and the Webster’s formulas to compute optimal cycle length based on real-time traffic conditions. \cite{saha2017delay} developed a modified delay estimation model based on the HCM framework, using data from seven four-legged signalized intersections. This model estimates the total and the average delay through queue length analysis and numerical integration.

Nevertheless, applying these analytical models to mixed traffic analysis is not straightforward, primarily due to the distinct driving behaviors between CAVs and HDVs, with special consideration for the leading effects of CAVs \citep{jiang2017eco}. Unlike HDVs, which respond reactively to signal changes and often stop at intersections, CAVs can coordinate with traffic signals via V2X communication to adjust their approach trajectories. This vehicle-signal coordination enables smoother, stop-free arrivals and allows CAVs to maintain tighter headways, enhancing both mobility and capacity. Compared to HDV leaders, CAV leaders can leverage advanced control strategies, V2X communication, and precise actuation to optimize traffic flow, especially at signalized intersections. Consequently, transportation agencies need to account for the impact of mixed traffic flow when estimating delays at signalized intersections to support informed decision-making and effective signal timing planning (e.g., optimal cycle length).

To advance intersection signal design in the era of mixed traffic, there is a need for a simple yet effective stochastic analytical framework. This paper proposes an approach to approximate delays and determine the optimal cycle length at an isolated signalized intersection, accounting for the distinct behaviors of both CAVs and HDVs. Specifically, the framework estimates traffic delay through a traffic composition analysis enabled by a Markov chain model. By capturing differences in arrival behavior and headway characteristics between CAVs and HDVs, mixed traffic delays are analytically derived using queuing theory. The optimal cycle length is then obtained by differentiating the overall delay function with respect to the cycle length. Finally, numerical experiments and sensitivity analyses are conducted to provide a comprehensive understanding of signal control performance under mixed traffic conditions.

The remainder of this paper is organized as follows: \textbf{Section 2} introduces the assumption and proposed stochastic capacity analytical framework under mixed traffic flow. \textbf{Section 3} presents the approximated delay model for an isolated intersection with mixed traffic flow and its optimal cycle length. \textbf{Section 4} conducts numerical experiments, highlighting the impact of multiple factors on control delay. Finally, \textbf{Section 5} gives the conclusion and directions for future research.

\section{CAPACITY DERIVATION FOR MIXED TRAFFIC}
\label{sec:capacity-derivation}
In CAV systems, multi-leader topology (MLT) in V2V communication is widely used for car-following control, as it allows each vehicle to utilize information from multiple predecessors to smooth traffic and ensure string stability \citep{wang2020cooperative, bian2019reducing, zhou2019robust, yue2024hybrid, li2024enhancing}. Therefore, MLT is adopted in this study to support further analysis. In this framework, the ego CAV receives traffic states (spacing, velocity, acceleration) from its $n$ preceding CAVs and transmits its own state to $n$ following CAVs, enabling localized coordination within the CAV subset of the platoon.

This V2V communication enables a shorter headway when multiple consecutive CAVs appear in a platoon, thereby improving overall traffic capacity \cite{chen2023stochastic}. However, when an HDV appears between CAVs, it breaks the chain of information sharing and disrupts coordination, often resulting in a longer headway. As a result, the headways in mixed traffic are not uniform, but instead depend heavily on the specific sequence and grouping of CAVs and HDVs within a platoon. To capture this stochastic dependency, we model the vehicle-type sequence using a Discrete Time Markov Chain (DTMC). Based on this formulation, we compute the steady-state probabilities of different headway configurations associated with various platoon compositions. These probabilities are then used to estimate the long-run expected mixed traffic capacity, forming the analytical basis for the subsequent intersection delay approximation.

\subsection{Steady State Mixed Traffic Sequence Analysis}
The steady state is a state in which the probabilities of transitioning between different states no longer change over time. In our case, we consider a state space consisting of various lengths of CAV platoons. Hence, steady-state probability refers to the probabilities of each platoon type at the intersection when the whole system reaches equilibrium. Figure \ref{fig:perf} illustrates the recurrent state transitions between different types of mixed traffic platoons, depicting the penetration rates ($p$) of the CAVs as the probability that the next follower will be a CAV. 

\begin{figure}[ht]
\centering
    \includegraphics[width=0.45\textwidth]{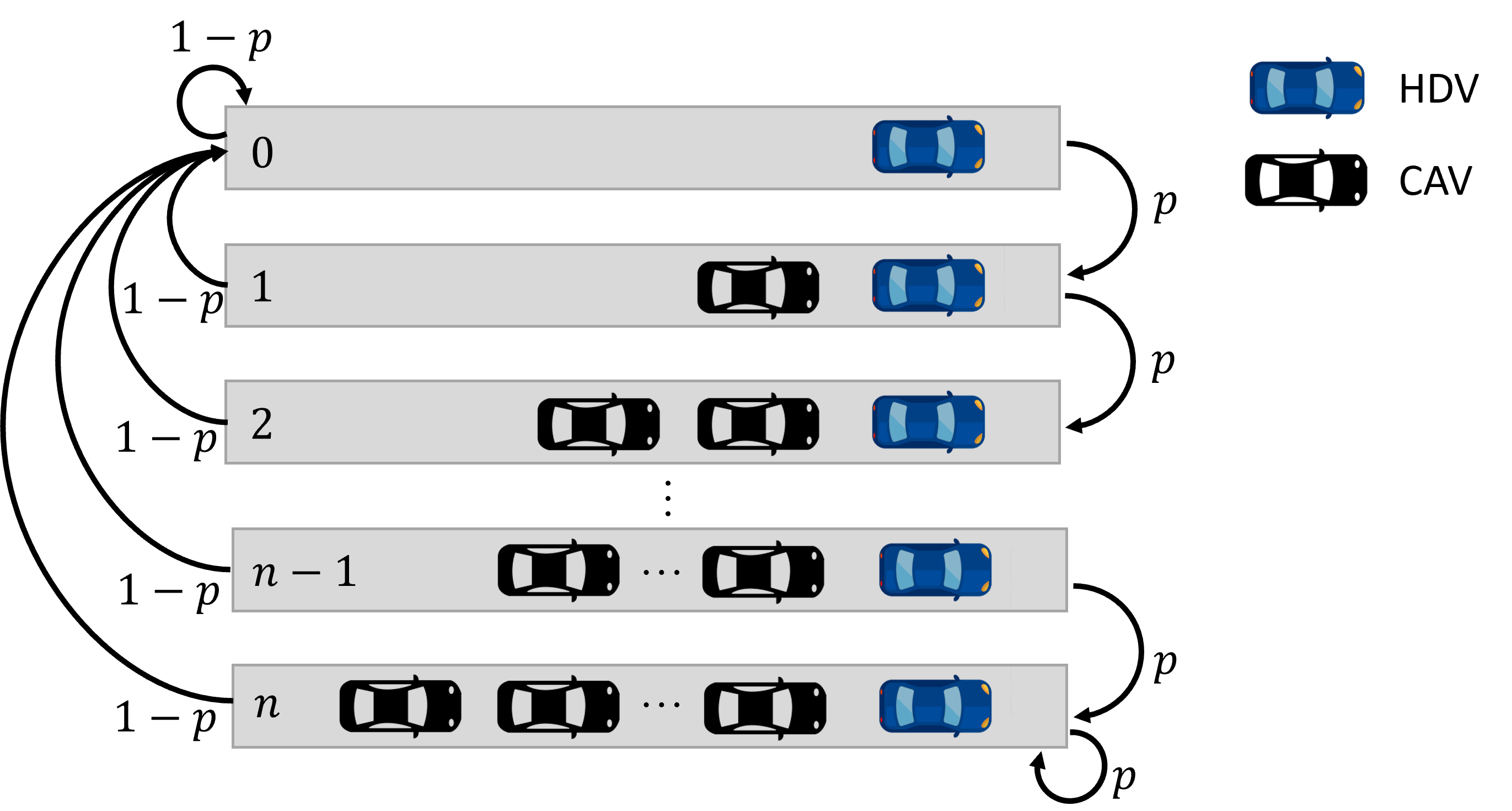}
    \caption{Illustative example of DTMC for mixed traffic platoon}\label{fig:perf}
\end{figure}

The DTMC used to model platoon patterns consists of states $n$, where the state number denotes the value of successive CAVs in the platoon and $n$ represents the maximum vehicle communication capacity within the communication range. Thus, the formula for the state space $S$, $ S = \{0, 1, 2, \ldots, i, \ldots, n-1, n\} $, and the transition probability matrix $P = [P_{ij}] \in \mathbb{R}^{(n+1) \times (n+1)}$, $P_{ij}$ is the transition probability from states $i$ to states $j$

\begin{equation}
 P_{ij} =
\begin{cases}
p & \text{if } j = i + 1, j = i = n \\
1-p & \text{if } j = 0 \\
0 & \text{otherwise}
\end{cases}
\label{eq:transition_matrix}
\end{equation}

Then, we can denote $\pi$  as long-run (steady-state) probabilities to describe the probability of each state when the DMTC reaches equilibrium: 

\begin{equation}
\pi = [\pi_0, \pi_1, \pi_2, \ldots, \pi_i, \ldots, \pi_{n-1}, \pi_n]
\label{eq:pi_vector}
\end{equation}
where $\pi_i$ denotes the steady-state probabilities of each state. 

The long-run or steady-state properties of the DTMC for a transition probability with $i$ rows and $j$ columns are defined as follows: 

\begin{equation}
    \pi = \pi P \\
    \label{eq:steady_state}
\end{equation}
\begin{equation}
    \sum_{i=0}^{n} \pi_i = 1
    \label{eq:sum_constraint}
\end{equation}

The distribution of vehicles within a platoon directly influences the time gap between consecutive vehicles, as variations in vehicle placement impact the time gaps and synchronization needed to maintain safe and efficient traffic flow. With Eqs.~\eqref{eq:transition_matrix}--\eqref{eq:sum_constraint}, we can use DTMC to determine the steady-state probability of time gaps under different vehicular distributions in the platoon, which correspondingly builds the base to solve the long-run expected capacity of the intersection, which is shown below:

\begin{equation}
\pi_i = \frac{p^i}{\frac{p^n}{1-p} + \sum_{m=0}^{n-1}p^m}
\label{eq:pi_closed_form}
\end{equation}\textbf{}

\subsection{Capacity Analysis}
Based on the steady-state probability of different vehicular distribution in the platoon, we can derive the expected time gap $E[\tau]$ for mixed traffic.

\begin{equation}
E[\tau] = \pi_0 \tau_{HDV} + \sum_{i=1}^{n} \pi_i \tau_{CAV,i}
\label{eq:expected_gap}
\end{equation}
Where $\tau_{\text{HDV}}$ is the desired time gap for HDVs and $\tau_{{CAV}, i}$ is the desired time gap for $i$ successive CAVs in the platoon.

Unlike traditional analysis, which treats all of the desired time gap for different successive CAVs in the platoon are the same $\tau_{{CAV}, i} \equiv \tau_{CAV}$, here we use the string stability criteria proposed by \cite{chen2023stochastic} to derive the minimum time gap between CAVs, since capacity is defined as the maximum sustainable traffic throughput. The desired time gap for CAV $\tau_{CAV, i}$ follows:

\begin{equation}
\tau_{\text{CAV},i} = \max \left( \tau_{safe}, \frac{4\omega_{v}}{\omega_{e}\cdot(1+i)} \right)
\label{eq:min_gap}
\end{equation}
Where $\omega_{e}$ is the deviation from equilibrium spacing feedback gains, $\omega_{v}$ denotes speed difference feedback gains. $\tau_{safe}$ is a predefined value of the time gap $\tau$ due to safety concerns.

By the relationship between the car-following time gap and traffic capacity $c$, we can have the following:

\begin{equation}
c = \frac{1}{E[\tau] + \frac{L}{v_{free}}}
\label{eq:capacity_eq}
\end{equation}
Where $E[\tau]$ represents the expected desired traffic time gap, $L$ denotes the average length of the vehicle, and $v_{\text{free}}$ is the free flow speed.

Given Eqs.~\eqref{eq:pi_closed_form}--\eqref{eq:capacity_eq}, we derived the expected capacity for mixed traffic system as:
\begin{equation}
c = \frac{1}{
    \frac{\tau_{HDV}} {\frac{p^n}{1-p} + \sum_{m=0}^{n-1} p^m}+ \frac{L}{v_{free}} +
     \sum_{i=1}^{n} \frac{p^i \tau_{CAV,i}}{\frac{p^n}{1-p}+ \sum_{m=0}^{n-1} p^m}
}
\end{equation}

\section{DELAY APPROXIMATIONS AND ANALYTICAL MODEL}

As traffic volume approaches the intersection's capacity, the intersection delay increases. When demand exceeds capacity, excessive queues form, leading to significant delays and even gridlock. This section introduces a delay approximation method based on queuing theory that considers the behavioral differences between CAVs and HDVs to develop an analytical delay model for the mixed-traffic intersection, building based on the previous analysis of mixed traffic capacity.

\subsection{Intersection arrival \& departure behaviors and delay approximation}
CAVs and HDVs exhibit different arrival and departure behaviors at isolated signalized intersections. As shown in Figures \ref{fig:ADCAV} (a) and (b), CAVs can communicate with the intersection signal to anticipate upcoming speed reductions, enabling them to adjust their speed in advance and accelerate strategically to pass through the intersection at free-flow speed without stopping or compromising the effectiveness of the green phase.
\begin{center}
  \includegraphics[width=0.45\textwidth]{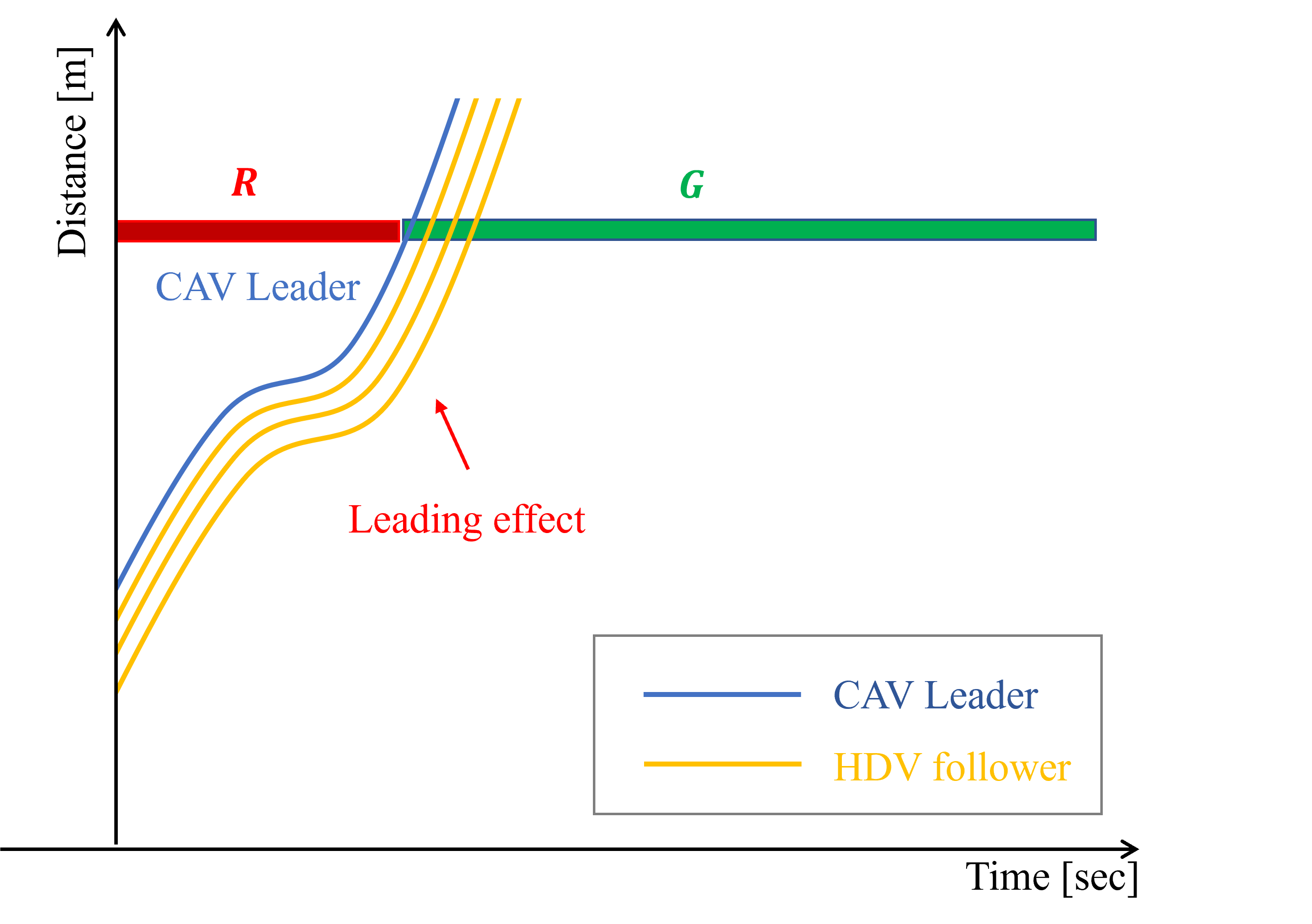}
  
  \textbf{(a)} Arrival behavior of CAV-led platoon
  \vspace{-6pt}

  \includegraphics[width=0.45\textwidth]{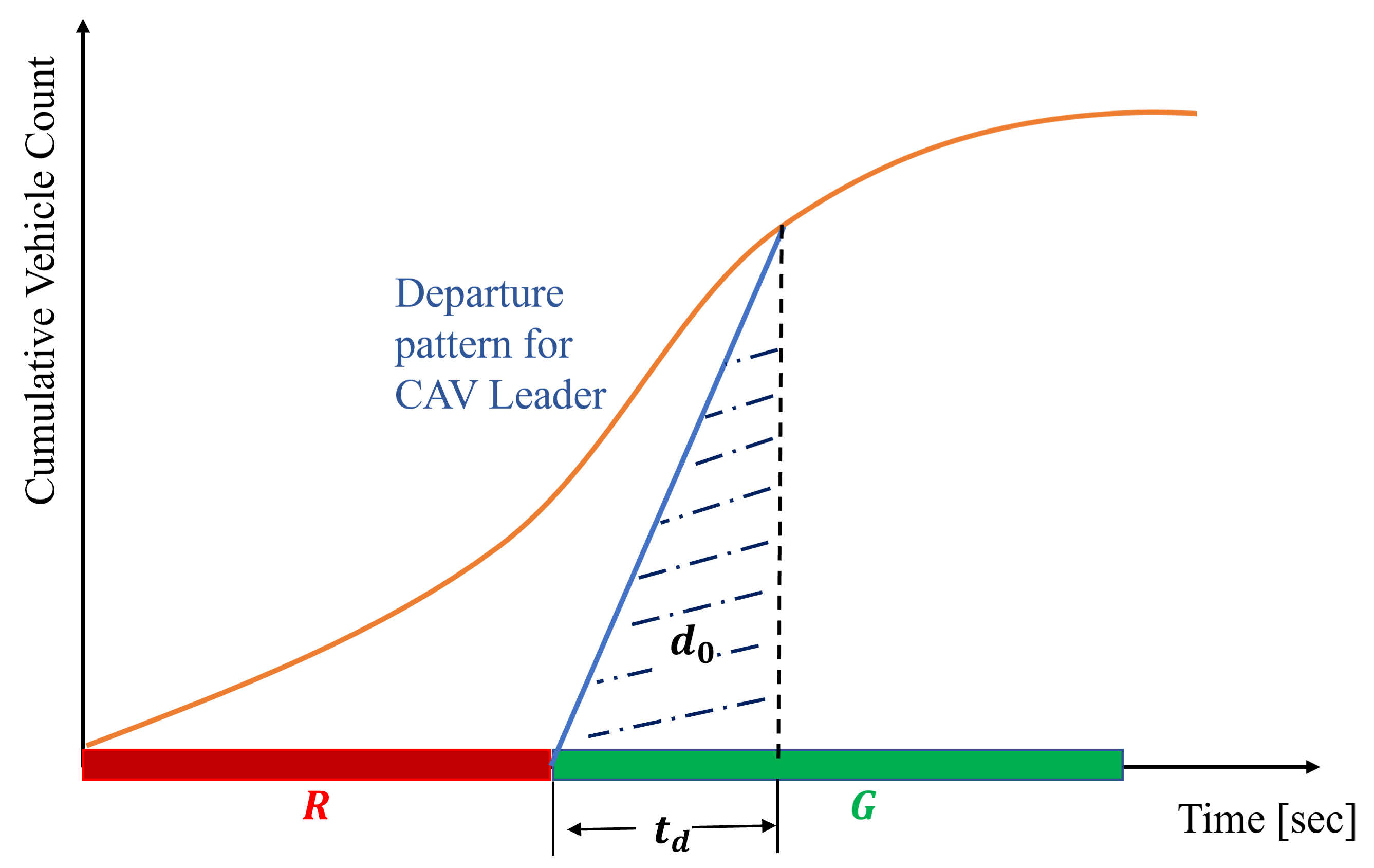}
  
  \textbf{(b)} Departure behavior of CAV-led platoon
  \vspace{-6pt}
  \captionof{figure}{Arrival and departure behavior of CAV-led platoon.}
  \label{fig:ADCAV}
\end{center}

In contrast, as shown in Figures \ref{fig:ADHDV} (a) and (b),  HDVs rely solely on visual cues and react to the signal state in real-time. When approaching a red light, HDVs gradually slow down and come to a complete stop at the intersection. When the signal turns green, human drivers always have time lost due to their reaction time and acceleration time to reach the free-flow speed. Comparing the departure behaviors of the CAV-led platoon with the HDV-led platoon (Figure \ref{fig:ADCAV} (b) and Figure \ref{fig:ADHDV} (b)), we can observe that the CAV-led platoon exhibits lower control delay due to their ability to coordinate with the intersection, highlighting the advantages of CAVs in reducing delays at the intersection. 

\begin{center}
  \includegraphics[width=0.45\textwidth]{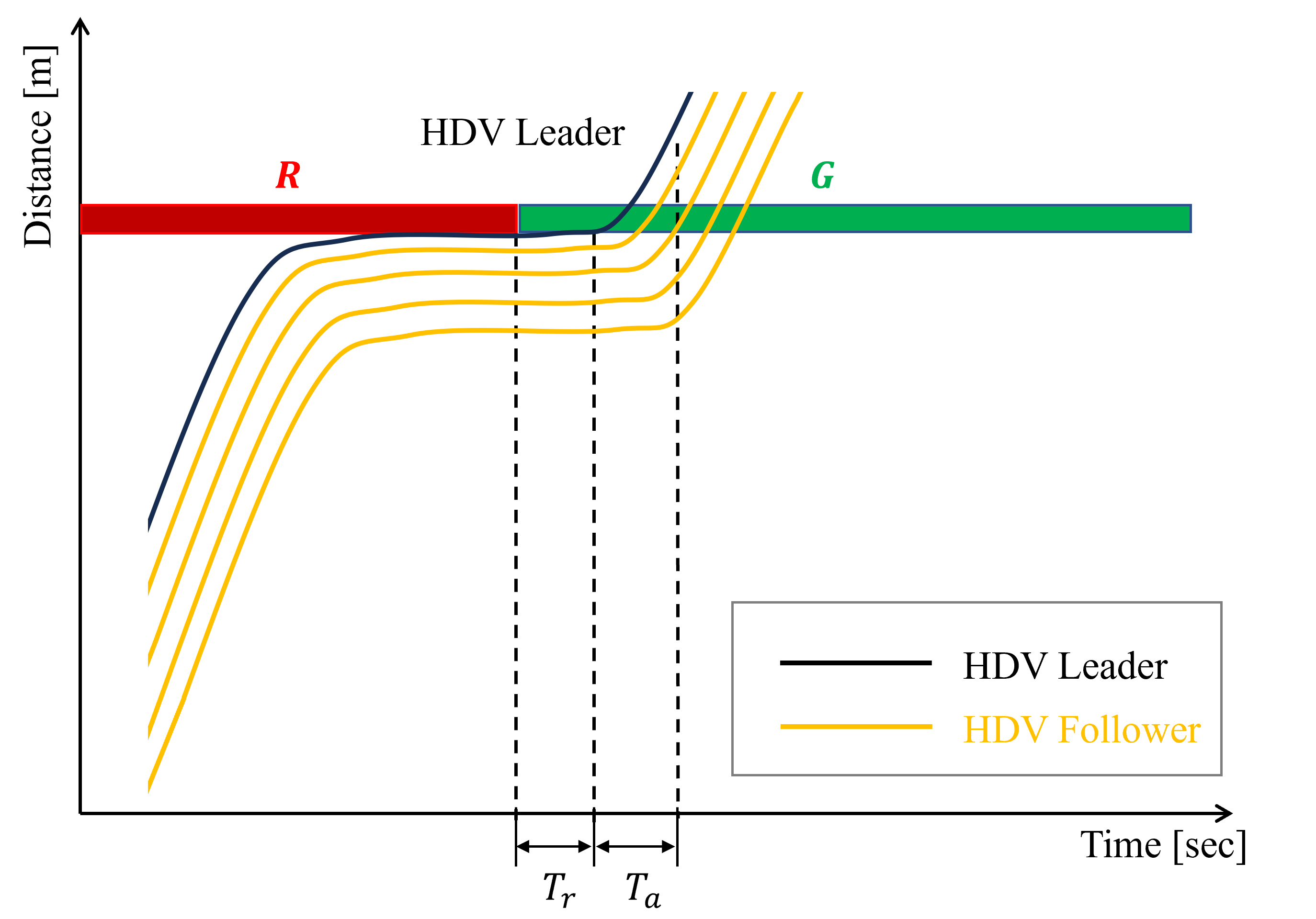}
  
  \textbf{(a)} Arrival behavior of HDV-led platoon

  \vspace{-6pt}

  \includegraphics[width=0.45\textwidth]{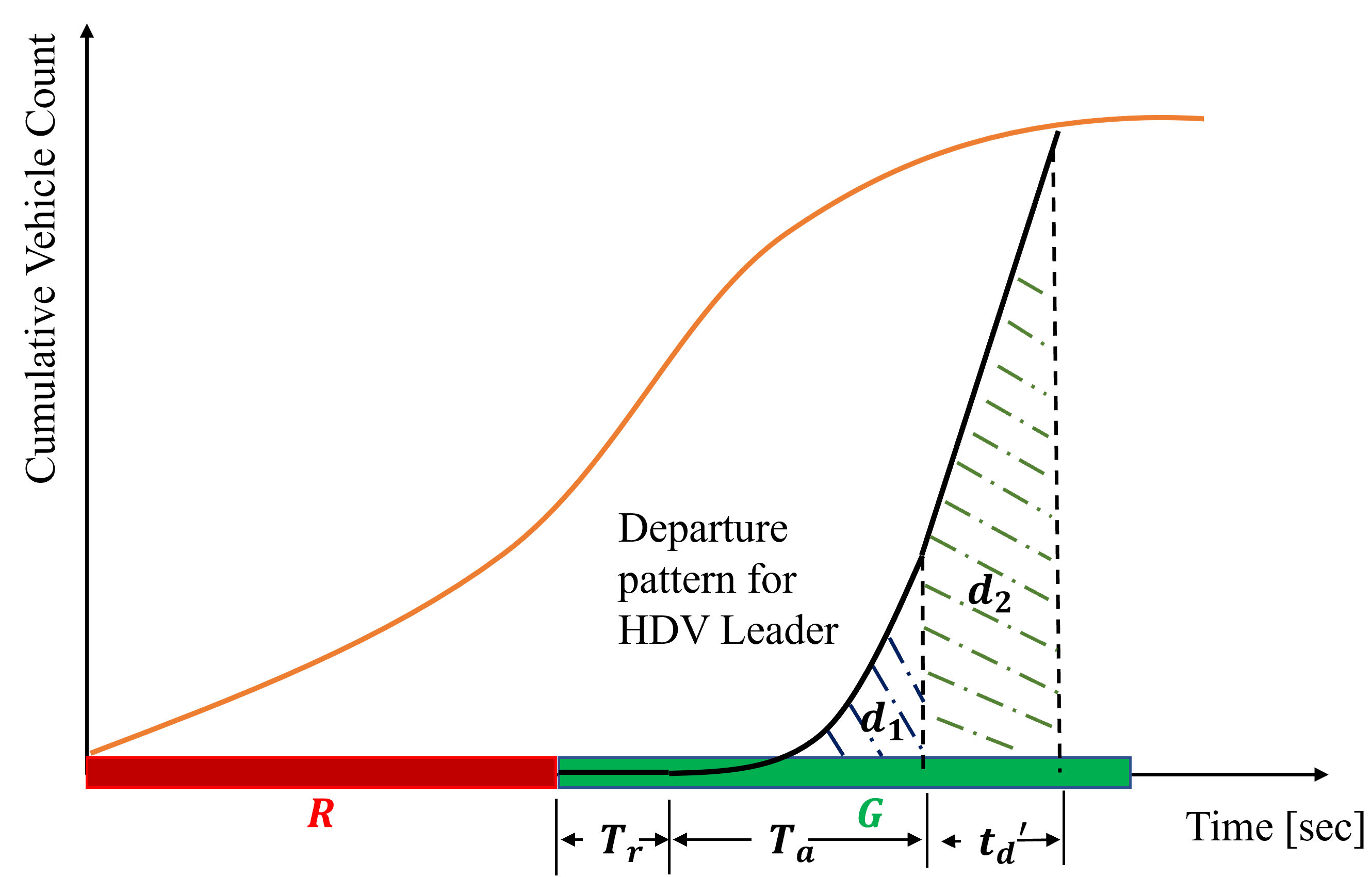}
  
  \textbf{(b)} Departure behavior of HDV-led platoon
  \vspace{-6pt}
  \captionof{figure}{Arrival and departure behavior of HDV-led platoon.}
  \label{fig:ADHDV}
\end{center}

To approximate the delay of these two departure behaviors, the HDV acceleration phase is modeled as a quadratic function \citep{zhou2017rolling}, where HDV-led platoon can reach the capacity $c$ (veh/$s$) after the acceleration process time $T_a$. The intersection parameters include $R$ (red time), $G$ (green time), $T_r$ (HDV reaction time).

According to queuing theory, the total delay at the intersection during a cycle for the departure behavior of CAV-led platoon, denoted as $D_\text{CAV}$, is given by the area between the cumulative arrival curve $q(t)$ and the departure curve:

\begin{equation}
D_{\text{CAV}} = \int_{0}^{R+t_d} q(t) \, dt - \frac{1}{2} t_d c^2
\label{eq:DCAV}
\end{equation}

where $t_d$ is the time that the queue dissipates in the departure behavior of a CAV-led platoon, which satisfies that: 

\begin{equation}
q(R+t_d)=c \cdot t_d
\end{equation}

Based on the approximated delay model and under the assumption that vehicle acceleration follows a quadratic profile, the cumulative departure curve of the HDV-led platoon can be represented as a piecewise function encompassing three distinct phases. These phases correspond to the typical HDV departure behavior at a signalized intersection:
\begin{itemize}
    \item \textbf{Phase 1} spans from $t = 0$ to $t = R + T_r$, during which no vehicles depart due to the red signal and the human driver's reaction time.
    \item \textbf{Phase 2} covers the interval from $t = R + T_r$ to $t = R + T_r + T_a$, representing the acceleration phase of the HDV-led platoon. During this phase, the departure rate increases nonlinearly, modeled using a quadratic function.
    \item \textbf{Phase 3} begins at $t = R + T_r + T_a$ and continues until $t = R + T_r + T_a + t_d'$, where the departure rate stabilizes at the saturation flow rate as the platoon reaches free-flow speed.
\end{itemize}
The cumulative departure function capturing these phases is presented below.

\begin{equation}
\resizebox{1\hsize}{!}{$
N(t) =
\begin{cases}
0 & t \in [0, R+T_r] \\
\frac{c}{2T_a} t^2 - \frac{c(R+T_r)}{T_a} t + \frac{c}{2T_a} (R^2 + 2RT_r + T_r^2) & t \in [R+T_r, R+T_r+T_a] \\
ct - c \left( \frac{T_a}{2} + R + T_r \right) & t \in [R+T_r+T_a, R+T_r+T_a + t_d']
\end{cases}
$}
\end{equation}
where $N(t)$ is the cumulative departure vehicles of the HDV-led departure platoons at the time $t$.

Therefore, similar to Eq.~\ref{eq:DCAV}, the total delay at the intersection during a cycle for the departure characteristics of a HDV-led platoon, denoted as $D_\text{HDV}$, is given by the area between the cumulative arrival curve $q(t)$ and the corresponding HDV departure curve:
\begin{equation}
\resizebox{\linewidth}{!}{$
\begin{split}
D_\text{HDV} = & \int_{0}^{R+T_r+T_a+t_d'} q(t) dt \\
& - \int_{R+T_r}^{R+T_r+T_a} \left( \frac{c}{2T_a} t^2 - \frac{c(R+T_r)}{T_a} t + \frac{c}{2T_a} (R^2 + 2T_r R + T_r^2) \right) dt \\
& - \int_{R+T_r+T_a}^{R+T_r+T_a+t_d'} \left( ct - c\left(\frac{T_a}{2} + R + T_r\right) \right) dt \\
= & \int_{0}^{R+T_r+T_a+t_d'} q(t)dt - \frac{cT_a^2}{6} - \frac{t_d' c}{2T} \left( \frac{c}{2T_a} + q(R+T_r+T_a+t_d') \right)
\end{split}
$}
\end{equation}
where $t_{d}^{'}$ is the time that the queue diminishes, following: 

\begin{equation}
q(R+T_r+T_a+t_{d}^{'} )=c\cdot t_{d}^{'}+  \frac{c}{2T_a}
\end{equation}


\subsection{Analytical delay models of the intersection}

To provide a direct understanding of the impact of CAVs on the signalized intersection delay while ensuring the analytical form of delay, we considered the curve $q(t)$ as a constant arrival rate $\bar{q}$ for simplification. Figure \ref{fig:4} (a) \& (b) present the different departure behavior of the CAV-led and the HDV-led platoon at the isolated signalized intersection.

\begin{center}
  \includegraphics[width=0.44\textwidth]{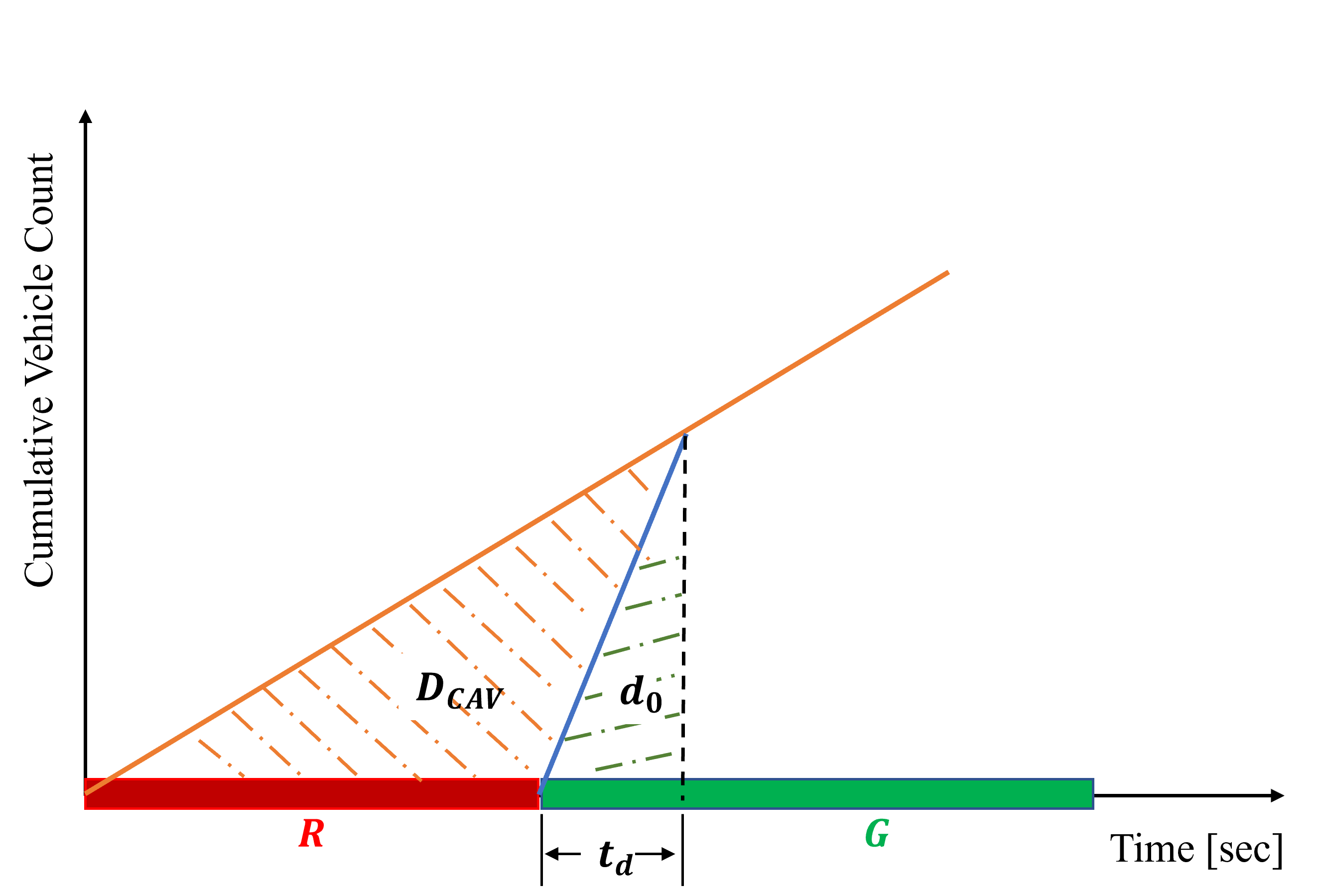}
  
  \textbf{(a)} Delay model for CAVs

  \vspace{1em}

  \includegraphics[width=0.4\textwidth]{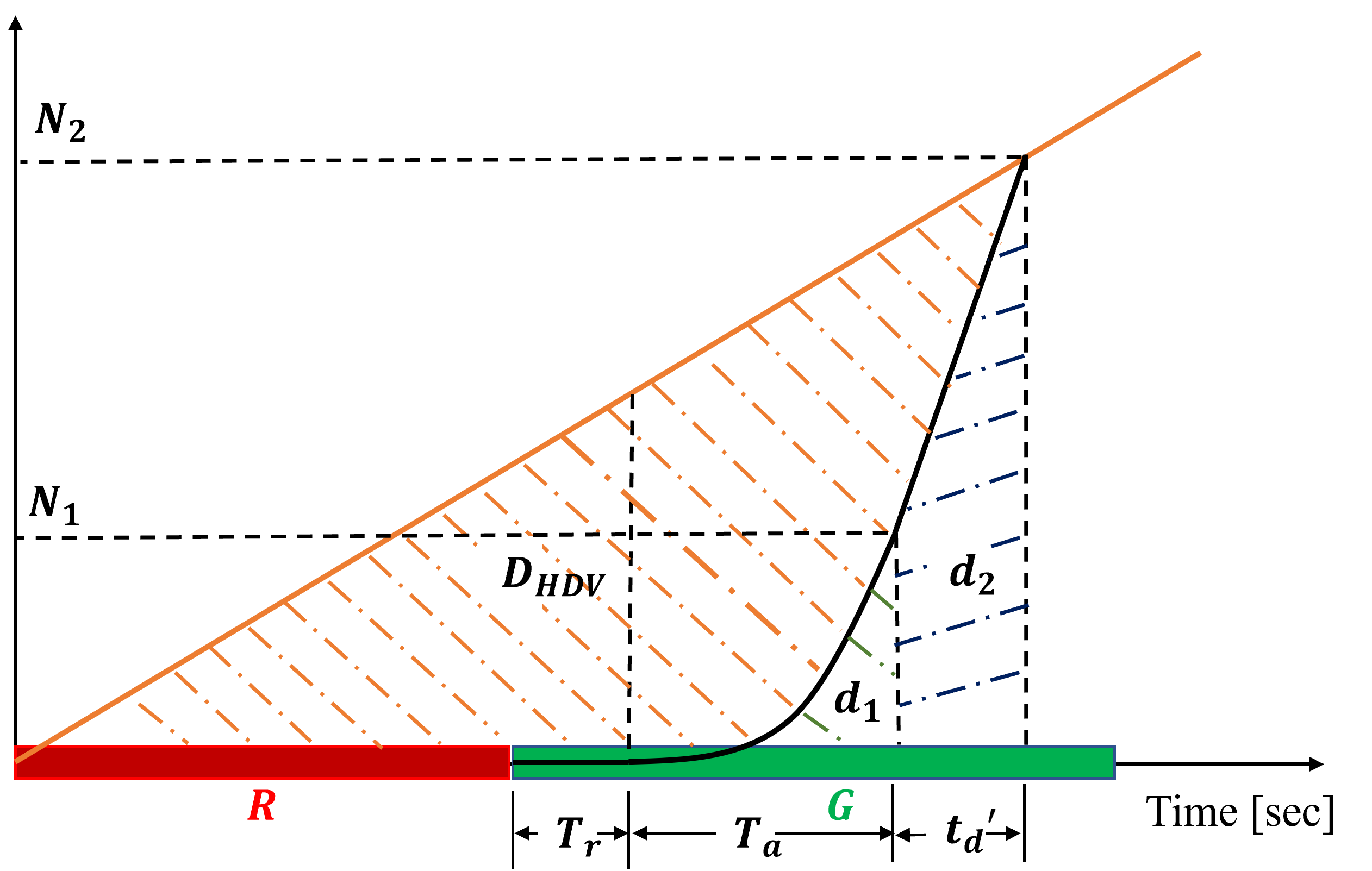}
  
  \textbf{(b)} Delay model for HDVs

  \captionof{figure}{Arrival and departure behavior of an HDV-led platoon.}
  \label{fig:4}
\end{center}

Specifically, the total delay of the CAV-led platoon is represented by the area $D_\text{CAV}$ in Figure~\ref{fig:4} (a), while the total delay of the HDV-led platoon corresponds to the area $D_\text{HDV}$ in Figure~\ref{fig:4} (b). In these illustrations, the shaded area $d_0$ represents the total free-flow delay experienced by the CAV-led platoon, whereas the combined area $d_1 + d_2$ indicates the total free-flow delay for the HDV-led platoon. The equations of $D_\text{CAV}$ and $D_\text{HDV}$ are shown below:

\begin{equation}
\label{eq:distance_models}
\begin{aligned}
D_\text{CAV} &= \frac{1}{2} q (R + t_d)^2 - d_0, \\
D_\text{HDV} &= \frac{1}{2} q (R + T_a + T_r + t_d')^2 - d_1 - d_2
\end{aligned}
\end{equation}

After formulating the delay models, the total delay for a cycle's movement is computed. It should be noted that over-saturated conditions are not considered in this analysis. 

For CAV delay \( D_\text{CAV} \): If the delay \( t_d \) is less than the green time \( G \) (indicating under-saturation), as shown in Figure ~\ref{fig:4} (a)). The average flow rate $\bar{q}$ at time $R + t_d$ is given by $\bar{q}(R + t_d) = c\cdot t_d$. This leads to the following relation for $t_d$:

\begin{equation}
\label{eq:td}
t_d = \frac{\bar{q}\cdot R}{c - \bar{q}}
\end{equation}
The solution exists when the system operates under under-saturation, that is, when \(c > \bar{q}\). The total number of vehicles $N_{\text{total}}$ can be computed as the flow rate times the sum of red time and green time:

\begin{equation}
\label{eq:N}
 N_{\text{total}} = \bar{q}(R + G) 
\end{equation}

Given \cref{eq:distance_models,eq:td,eq:N}, we can derive the total delay ($D_\text{CAV}$) and the average delay ($\hat{D}_\text{CAV}$) of the departure behavior of the CAV-led platoon for the movement in a cycle. The delay for CAV-led platoon (\( D_\text{CAV} \)) is given by:

\begin{equation}
D_{\text{CAV}} = \frac{\bar{q}(R + t_d)^2}{2} - \frac{c t_d^2}{2} = \frac{c\bar{q}R^2}{2(c-\bar{q})}
\end{equation}

The average delay for CAV-lead platoon, \( \hat{D}_\text{CAV} \) is then:

\begin{equation}
\hat{D}_\text{CAV} = \frac{D_{\text{CAV}}}{N_\text{total}} = \frac{R^2}{2(R + G)} + \frac{\bar{q} R^2}{2(c - \bar{q})(R + G)}
\end{equation}

For HDV, if \( T_r + T_a + t_d' < G \) (indicating under-saturation, as shown in Figure~\ref{fig:4} (b), then we have:
\begin{equation}
\label{eq:N1N2}
\left\{
\begin{aligned}
N_1 &= - \frac{c}{T_a} \left(R + T_r\right)\left(R + T_r + T_a\right) \\
& +\frac{c}{2T_a} \left(R + T_r + T_a\right)^2 + \frac{c}{2T_a} \left(R + T_r\right)^2 = \frac{c\,T_a}{2}, \\[1ex]
N_2 &= \bar{q}\left(R + T_r + T_a + t_d'\right).
\end{aligned}
\right.
\end{equation}
where $N_{1}$ denotes the number of vehicles crossing the intersection in the acceleration phase, and $N_{2}$ is the sum of the total number of vehicles crossing the intersection at the acceleration phase and the free flow speed phase.

By geometry, $N_{1}$ and $N_{2}$ satisfy that:
\begin{equation}
\label{eq:N1N2rela}
N_{1}+c\cdot t_{d}^{'}=N_{2}
\end{equation}

Given Eqs. \ref{eq:N1N2} \& \ref{eq:N1N2rela}, the time that the queue diminishes when the green light starts $t_{d}^{'}$ is derived as follows:
\begin{equation}
\label{eq:tdprime}
t_{d}' = \frac{\bar{q}(R+T_r+T_a) - c/2\cdot T_a}{c - \bar{q}}
\end{equation}

Therefore, $N_2$ can be further derived as:
\begin{equation}
\label{eq:N2}
N_2 = \frac{c\cdot T_a}{2} + \frac{\bar{q}(R+T_r)+T_a(\bar{q} - c/2)}{c - \bar{q}} \cdot c
\end{equation}

Given Eqs. \ref{eq:distance_models}, \ref{eq:N1N2}, \ref{eq:N1N2rela}, \ref{eq:tdprime}, and \ref{eq:N2}, the total delay ($D_\text{HDV}$) and the average delay of the HDV-led platoon ($\bar{D}_\text{HDV}$) for the movement in a cycle is derived as follows:  

\begin{equation}
\scalebox{0.75}{$
\begin{aligned}
D_{\text{HDV}} 
&= \frac{1}{2} \bar{q}(R + T_r + T_a + t_d')^2 
    - \frac{1}{2} (N_1 + N_2) t_d'^2 \\
&\quad - \int_{R+T_r}^{R+T_r+T_a} 
    \left( 
        \frac{c}{2T_a} t^2 
        - \frac{c(R+T_r)}{T_a} t 
        + \frac{c}{2T_a}(R^2 + 2T_r R + T_r^2) 
    \right) dt \\
&= \frac{\bar{q}}{2} 
    \left( 
        R + T_r + T_a 
        + \frac{\bar{q}(R+T_r) + T_a(\bar{q} - 0.5c)}{c - \bar{q}} 
    \right)^2 \\
&\quad - \frac{c}{2} 
    \left( 
        T_a + \frac{\bar{q}(R+T_r) + T_a(\bar{q} - 0.5c)}{c - \bar{q}} 
    \right)
    \left( 
        \frac{\bar{q}(R+T_r+T_a) - 0.5c T_a}{c - \bar{q}} 
    \right) \\
&\quad - \frac{c T_a^2}{6}
\end{aligned}
$}
\end{equation}


\begin{equation}
\scalebox{0.75}{$
\begin{aligned}
\bar{D}_{\text{HDV}} 
&= \frac{D_{\text{HDV}}}{N_{\text{total}}} \\[1ex]
&= \frac{1}{2(R+G)} 
    \left( 
        R + T_r + T_a 
        + \frac{ \bar{q}(R+T_r) + T_a\left( \bar{q} - \frac{c}{2} \right) }{ c - \bar{q} }
    \right)^2 \\[1ex]
&\quad 
    - \frac{c}{2\bar{q}(R+G)} 
    \left( 
        T_a 
        + \frac{ \bar{q}(R+T_r) + T_a\left( \bar{q} - \frac{c}{2} \right) }{ c - \bar{q} }
    \right) 
    \left( 
        \frac{ \bar{q}(R+T_r+T_a) - \frac{c}{2} T_a }{ c - \bar{q} }
    \right) \\[1ex]
&\quad 
    - \frac{c T_a^2}{6\bar{q}(R+G)}
\end{aligned}
$}
\end{equation}

Using the total probability theorem, the expected total delay $E[D]$ and expected average delay $E[\bar{D}]$ can be formulated as follows:

\begin{equation}
\scalebox{0.75}{$
\begin{aligned}
E[D] 
&= (1-p) D_{\text{HDV}} + p D_{\text{CAV}} \\[1ex]
&= \frac{\bar{q}(1-p)}{2} 
    \left( 
        R + T_r + T_a 
        + \frac{ \bar{q}(R+T_r) + T_a\left( \bar{q} - 0.5c \right) }{ c - \bar{q} }
    \right)^2 \\[1ex]
&\quad 
    - \frac{c(1-p)}{2} 
    \left( 
        T_a 
        + \frac{ \bar{q}(R+T_r) + T_a\left( \bar{q} - 0.5c \right) }{ c - \bar{q} }
    \right) 
    \left( 
        \frac{ \bar{q}(R+T_r+T_a) - 0.5c T_a }{ c - \bar{q} }
    \right) \\[1ex]
&\quad 
    - \frac{(1-p)c T_a^2}{6}
    + \frac{p c \bar{q} R^2}{2(c - \bar{q})}
\end{aligned}
$}
\end{equation}
Similarly, we have the average delay:

\begin{equation}
\scalebox{0.75}{$
\begin{aligned}
E[\bar{D}] 
&= (1-p) \hat{D}_{\text{HDV}} + p \hat{D}_{\text{CAV}} \\[1ex]
&= \frac{(1-p)}{2(R+G)} 
    \left( 
        R + T_r + T_a 
        + \frac{ \bar{q}(R+T_r) + T_a\left( \bar{q} - 0.5c \right) }{ c - \bar{q} }
    \right)^2 \\[1ex]
&\quad 
    - \frac{c(1-p)}{2\bar{q}(R+G)} 
    \left( 
        T_a 
        + \frac{ \bar{q}(R+T_r) + T_a\left( \bar{q} - 0.5c \right) }{ c - \bar{q} }
    \right) 
    \left( 
        \frac{ \bar{q}(R+T_r+T_a) - 0.5c T_a }{ c - \bar{q} }
    \right) \\[1ex]
&\quad 
    - \frac{(1-p)c T_a^2}{6\bar{q}(R+G)}
    + \frac{p c R^2}{2(c-\bar{q})(R+G)}
\end{aligned}
$}
\end{equation}

\subsection{Optimal cycle length}

Once the analytical delay equation is derived for a single approach of an isolated intersection, the overall delay of the isolated intersection is summed by the delays of each approach. 

\begin{equation}
D_{total} = \sum_{i=1}^{n} E[D_i]
\end{equation}

Then the red time of approach $i$ $(R_i)$ can be represented by:

\begin{equation}
R_i = (1 - \lambda_i) \cdot C
\end{equation}
where $C$ denotes the cycle length of the intersection, $\lambda_i$ is the proportion of the cycle which is effectively green (i.e. $G/C$).

Following Webster's methodology \citep{webster1966traffic}, by differentiating the equation for the total delay with respect to the cycle length ($C$) and setting the resultant derivative to zero, we can determine the optimal cycle length for the intersection with the minimum delay. The formula for the optimal cycle length ($C^*$) is presented below:

\begin{equation}
\resizebox{\linewidth}{!}{$
C^* = \frac{\sum_{i=1}^{n} \frac{c\bar{q_i}}{c-\bar{q_i}} \left[T_a(\lambda_i - \lambda_ip + p - 1) + 2T_r(\lambda_i - \lambda_ip + p - 1)\right]}{\sum_{i=1}^{n} \frac{c\bar{q_i}}{c-\bar{q_i}} (\lambda_i - 1)^2}
$}
\end{equation}

\begin{proposition}
Under our delay model, the total delay of mixed traffic with CAVs and HDVs at an isolated intersection is a monotonically increasing function of the cycle length $C$, where $C \in (0, \infty)$. Therefore, minimizing the cycle length reduces the total delay at the intersection.
\end{proposition}

\begin{proof}
The total delay function $D_{\text{total}}$, expressed as:


\begin{equation}
\scalebox{0.69}{$
\begin{aligned}
D_{\text{total}} 
&= \sum_{i=1}^{n} 
    \Bigg\{
    \frac{\bar{q}_i(1-p)}{2} 
    \left( 
        R_i + T_r + T_a 
        + \frac{ \bar{q}_i(R_i + T_r) + T_a\left( \bar{q}_i - 0.5c \right) }{ c - \bar{q}_i }
    \right)^2 \\[1ex]
&\quad
    - \frac{c(1-p)}{2} 
    \left( 
        T_a 
        + \frac{ \bar{q}_i(R_i + T_r) + T_a\left( \bar{q}_i - 0.5c \right) }{ c - \bar{q}_i }
    \right)
    \left( 
        \frac{ \bar{q}_i(R_i + T_r + T_a) - 0.5c T_a }{ c - \bar{q}_i }
    \right) \\[1ex]
&\quad
    - \frac{(1-p)c T_a^2}{6}
    + \frac{p c \bar{q}_i R_i^2}{2(c - \bar{q}_i)}
    \Bigg\}
\end{aligned}
$}
\end{equation}

Taking the derivative with respect to the cycle length \( C \):

\begin{equation}
\begin{aligned}
\frac{dD_{\text{total}}}{dC} &= D_{\text{total}}'(C) \\
&= \sum_{i=1}^{n} \frac{c \bar{q}_i}{c - \bar{q}_i} \Big( c(\lambda_i - 1)^2 + T_a(1 - \lambda_i)(1 - p) \\
&\quad + 2T_r(1 - \lambda_i)(1 - p) \Big)
\end{aligned}
\label{eq:dDtotal_dC}
\end{equation}
where the demand-to-capicity ratio of the $i-$th movement \( \lambda_i = \frac{\bar{q_i}}{c} \), $\bar{q_i}$ is the average arrival flow rate of the movement $i$ of the intersection. Given \( \lambda_i \leq 1 \) and \( p \leq 1 \), we have the summation are \textit{non-negative}, we conclude that:

\begin{equation}
\begin{aligned}
D_{\text{total}}'(C) \geq 0, \quad \forall C > 0.
\end{aligned}
\label{eq:ineq}
\end{equation}

Thus, \( D_{\text{total}}(C) \) is \textit{monotonically increasing} in \( C \), and minimizing \( C \) reduces the total delay. \qedhere
\end{proof}

In accordance with the Highway Capacity Manual (HCM, \cite{manual2000highway}),  the cycle length must exceed a minimum value $C_{\text{min}}$ to ensure sufficient time for all required clearance intervals (yellow and all-red phases) and to provide enough green time for critical traffic movements. Shorter cycles would not allow safe and efficient operation at signalized intersections.

\begin{equation}
C_{min} = \frac{E[L] * X_c}{X_{c} - \sum_{i=1}^{n} (\frac{v}{c})_{c_i}}
\end{equation}
where $E[L]$ is the expected total lost time of the intersection. $X_c$ measures the intersection's degree of saturation, representing the proportion of available capacity utilized by critical movements. $(\frac{v}{c})_{c_i}$ represents the flow ratio for critical lane group $i$, and $n$ is the number of the critical groups. 

According to the assumed departure behavior of an HDV-led platoon (illustrated in Figure \ref{fig:delay1}), only the departure platoons led by HDVs will experience the start-up lost time. This is because CAVs, due to their precise control, shorter reaction times, and coordination capabilities, are assumed to depart without incurring start-up delay \cite{zhong2019alternative}. However, regardless of the leading vehicle type of the departure platoon, the clearance lost time $L_c$, associated with the yellow and all-red intervals required for safe phase transitions, always exists and must be included in the calculation. Therefore, the expected total lost time can be calculated as: 
\begin{equation}
E[L] = (1-p)\cdot L_{s} + L_{c}
\end{equation}
where $L_{s}$ denotes as the total start-up lost time, $L_{c}$ represents the total clearance lost time. According to \cite{brown2000mutcd}, we set it to 4 seconds for our numerical experiments.

According to our delay model, the start-up delay of the HDV-led departure platoon can be shown in Figure \ref{fig:delay1}. Given Eq. \ref{eq:N1N2}, it can be calculated as:
\begin{equation}
L_{s} = \sum_{i=1}^{n} (T_{r,i} + T_{a,i} - \frac{N_{1,i}}{c_i}) = \sum_{n=1}^{n} (T_{r,i} + \frac{3T_{a,i}}{2})
\end{equation}
where $N$ denotes the total number of phases at the intersection, $i$ represents the phase number.

\begin{figure}[!h]
\centering
    \includegraphics[width=0.55\textwidth]{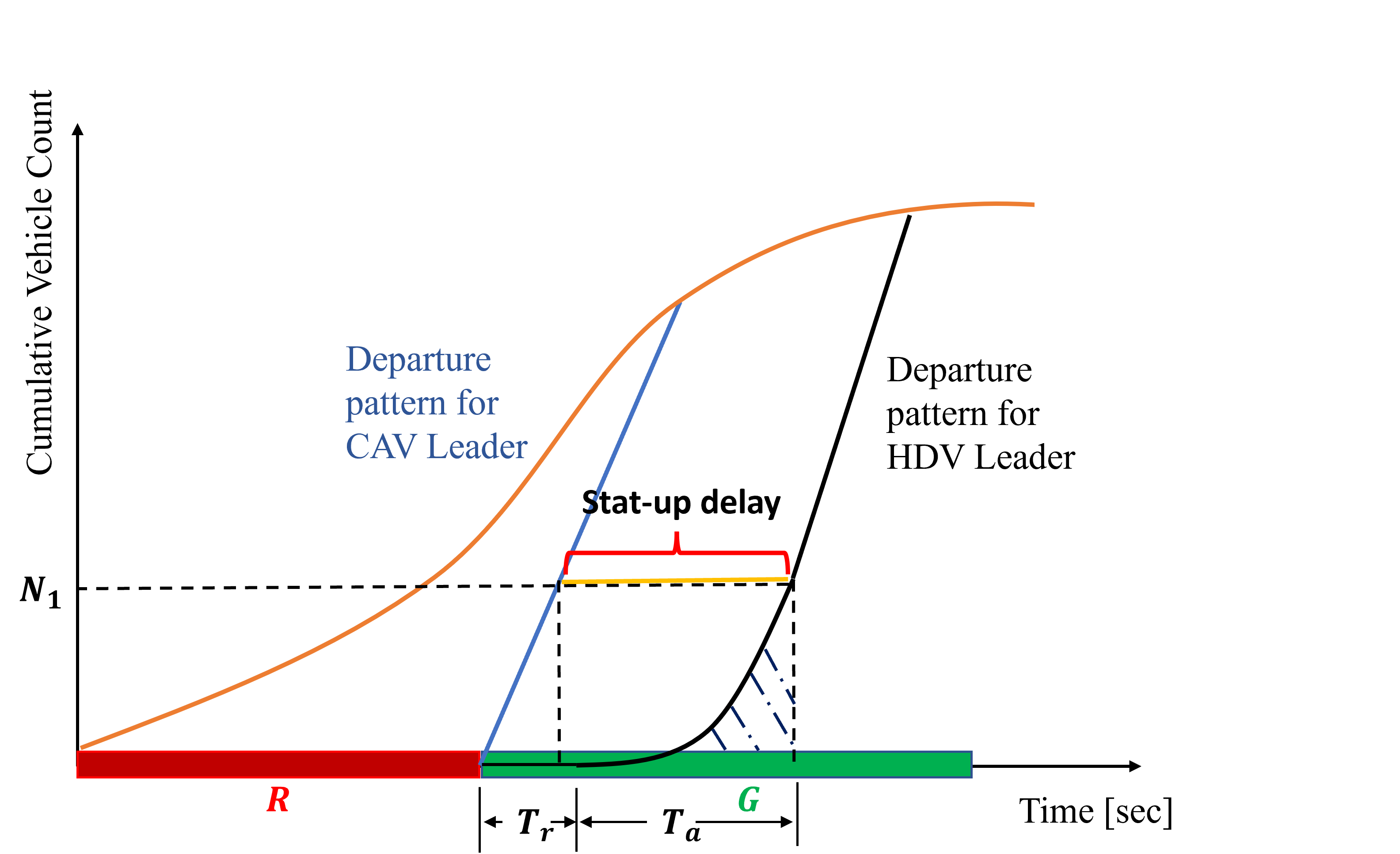}
    \caption{Illustration of start-up delay of the HDV-led platoon}\label{fig:delay1}
\end{figure}

\section{NUMERICAL EXPERIMENTS AND ANALYSIS}
In this section, we conducted several numerical experiments based on the previous derivation to better understand mixed traffic delays. Specifically, we performed a sensitivity analysis to obtain the relationship between the multiple parameters and the isolated intersection delay model under mixed traffic flow. Four key parameters are considered, including cycle length ($C$), CAV penetration rates ($p$), green ratio ($G/C$), and traffic arrival rate ($\bar{q}$), are considered. These parameters were chosen based on their critical role in intersection delay and capacity evaluation. The selection aligns with established methodologies in mixed traffic intersection simulation studies \citep{jiang2022dynamic}. We have divided the numerical experiments into four parts to systematically investigate the key factors influencing intersection performance: 1) the impact of arrival rate ($\bar{q}$), which directly affects the traffic demand level; 2) the impact of the green ratio ($G/C$), which determines the available service time within the cycle; 3) the joint impact of the green ratio and CAV market penetration rate ($G/C$ and $p$), to evaluate how traffic signal efficiency interacts with the increasing presence of automated vehicles; and 4) the joint impact of CAV market penetration rate and cycle length ($p$ and $C$) to assess how CAV adoption influences optimal signal timing design. The performance metric of our experiment is the vehicle average delay of the intersection $E[\bar{D}]$. 

\subsection{Numerical Experiments Settings}
The default value and feasible interval ranges for each parameter of the numerical experiments follow \citep{chen2023stochastic, naylor1997intersection}, which are listed in Table \ref{sample-table}.

\begin{table}[hbt!]
\centering
\caption{Default Value Setting for the Stochastic Capacity Sensitivity Analysis}
\label{sample-table}
\resizebox{\linewidth}{!}{
\begin{tabular}{lcl}
\toprule
\textbf{Parameter} & \textbf{Notation} & \textbf{Value} \\
\midrule
Deviation from Equilibrium Spacing Feedback Gains & \(\omega_{e}\) & 1.2 $\mathrm{sec}^{-2}$ \\
Speed Difference Feedback Gains & \(\omega_{v}\) & 0.5 $\mathrm{sec}^{-1}$ \\
Maximum Vehicle Communication Capacity & \(n\) & 5 vehs \\
Free-flow Speed & \(V_{\text{free}}\) & 15 m/sec \\
Reaction Time & \(T_r\) & 2 sec \\
Acceleration Process Time & \(T_a\) & 3 sec \\
Safe Time Gap & \(\tau_{\text{safe}}\) & 0.3 sec \\
Time Gap for HDV & \(\tau_{\text{HDV}}\) & 1.5 sec \\
Vehicle Length & \(L\) & 5 m \\
Penetration Rate of CAVs & \(p\) & (0, 1) \\
Cycle Length & \(C\) & (60 sec, 120 sec) \\
Green Ratio & \(G/C\) & (0.25, 0.75) \\
Arrival Flow Rate & \(\bar{q}\) & (0.15, 0.35) veh/sec \\
Intersection’s Degree of Saturation & \(X_{c}\) & 0.95 \\
\bottomrule
\end{tabular}
}
\end{table}

\subsection{Experiment results and analysis}
We performed an analysis to evaluate the influence of the traffic arrival rate $\bar{q}$. The mixed traffic arrival rate is set from 0.15 veh/$s$ to 0.35 veh/$s$ (i.e., 540vph to 1260vph) with an increment of 0.025. The expected average delay results for CAV market penetration rate and cycle length with different arrival rates $\bar{q}$ are shown in Figure \ref{fig:delay2} (left). We can conclude that as the traffic arrival rate $\bar{q}$ increases, the expected average delay increases correspondingly. 

We then conducted an analysis to examine the effects of the green ratio $G/C$ on the system. The range of the green ratio is from 0.25 $sec$ to 0.75 $sec$ with an increment of 0.025. Figure \ref{fig:delay2} (right) further shows the expected average delay results for CAV market penetration rate and cycle length with different green ratios ($G/C$). 
\begin{figure}[ht!]
\centering
    \includegraphics[width=0.45\textwidth]{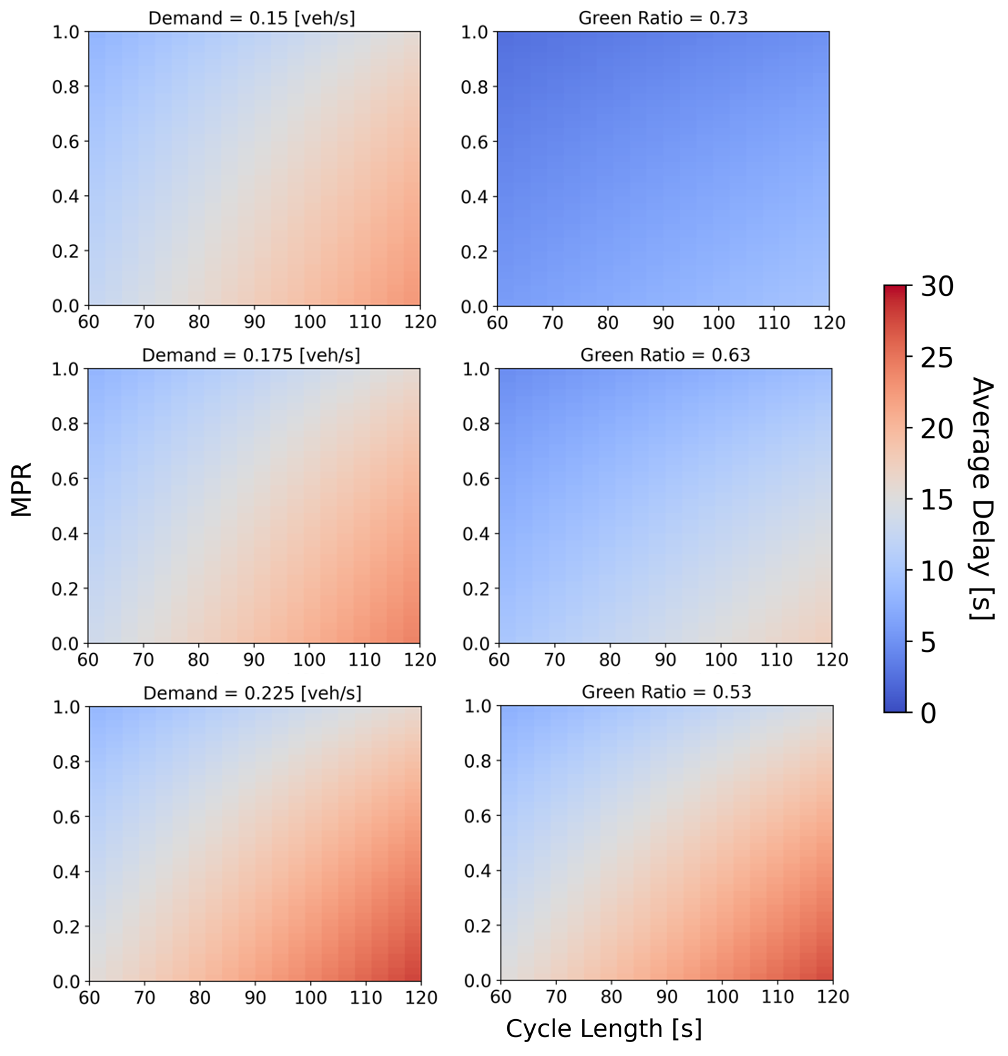}
    \caption{Expected average delay results for CAV market penetration rate and cycle length with different arrival rate $q$}\label{fig:delay2}
\end{figure}
Our findings indicate that the expected average delay decreases simultaneously as the green ratio increases. The Figure \ref{fig:delay2} further demonstrates that the green ratio has greater sensitivity to intersection delay compared to the market penetration rate. Specifically, comparing the right subfigure (impact of the green ratio) with the left subfigure (impact of CAV penetration rate) highlights that variations in the green ratio lead to more pronounced changes in the average delay. It is supported by \cite{hang2020modeling}'s study, which highlights that the green time ratio significantly influences traffic function reliability (TFR), which is closely related to intersection delay. It notes that under varying saturation levels, changes in the green time ratio have a substantial impact on TFR, whereas the effect of CAV penetration rates is comparatively less pronounced. \cite{wang2024traffic}'s study also emphasizes the immediate impact of signal timing adjustments on traffic flow. While increasing CAV penetration improves traffic efficiency, the benefits are more pronounced when combined with optimized signal timings. Research indicates that even with low CAV penetration rates, adjusting signal parameters such as the green ratio can lead to notable reductions in delays.


We conducted an additional analysis to assess the joint impact of CAV market penetration rate and cycle length. Assuming a green ratio of 0.55, the CAV market penetration rate $p$ is set from 0 to 1, and the range of the cycle length $C$ is from 60 $sec$ to 120 $sec$, with increments of 0.025 and 3, respectively. The left two plots of Figure \ref{fig:joint} indicate that the expected average delay increases when the CAV penetration rate $p$ is smaller and the cycle length $C$ is longer. This trend aligns with our expectations and previous studies \citep{li2018piecewise, ma2017parsimonious, zhou2017parsimonious, guo2019joint}, which suggest that CAVs can mitigate intersection delays and that higher market penetration rates lead to reduced delays. In addition, the right two plots of Figure \ref{fig:joint} show that for shorter cycle lengths (e.g., 60s, 70s), increasing $p$ still reduces delay, but the gap between different market penetration rate values is relatively small. As cycle length increases (e.g., 100s, 110s, 120s), the difference in delay reduction between low and high $p$ becomes more pronounced. It reveals that a longer cycle length amplifies the benefits of increasing the market penetration rate, further reducing the intersection control delay. Also, as the green ratio increase from 0.55 to 0.7, the delay reduction decreases significantly. As the green ratio increases, the marginal benefits of higher CAV penetration in reducing delays may diminish, indicating an optimal balance between signal timing and CAV integration.

To further analyze the impact of cycle length, we finally examine the joint influence of the CAV penetration rate and the green ratio on the system. Assuming a cycle length of 100 sec, the CAV market penetration rate is set from 0 to 1, and the $G$ over $C$ ratio range is from 0.25 to 0.75, both with increments of 0.025. Figure \ref{fig:joint} shows the expected average delay results of the CAV market penetration rate and $G$ over $C$ ratio. 
\begin{figure}[htbp!]
\centering
    \includegraphics[width=0.47\textwidth]{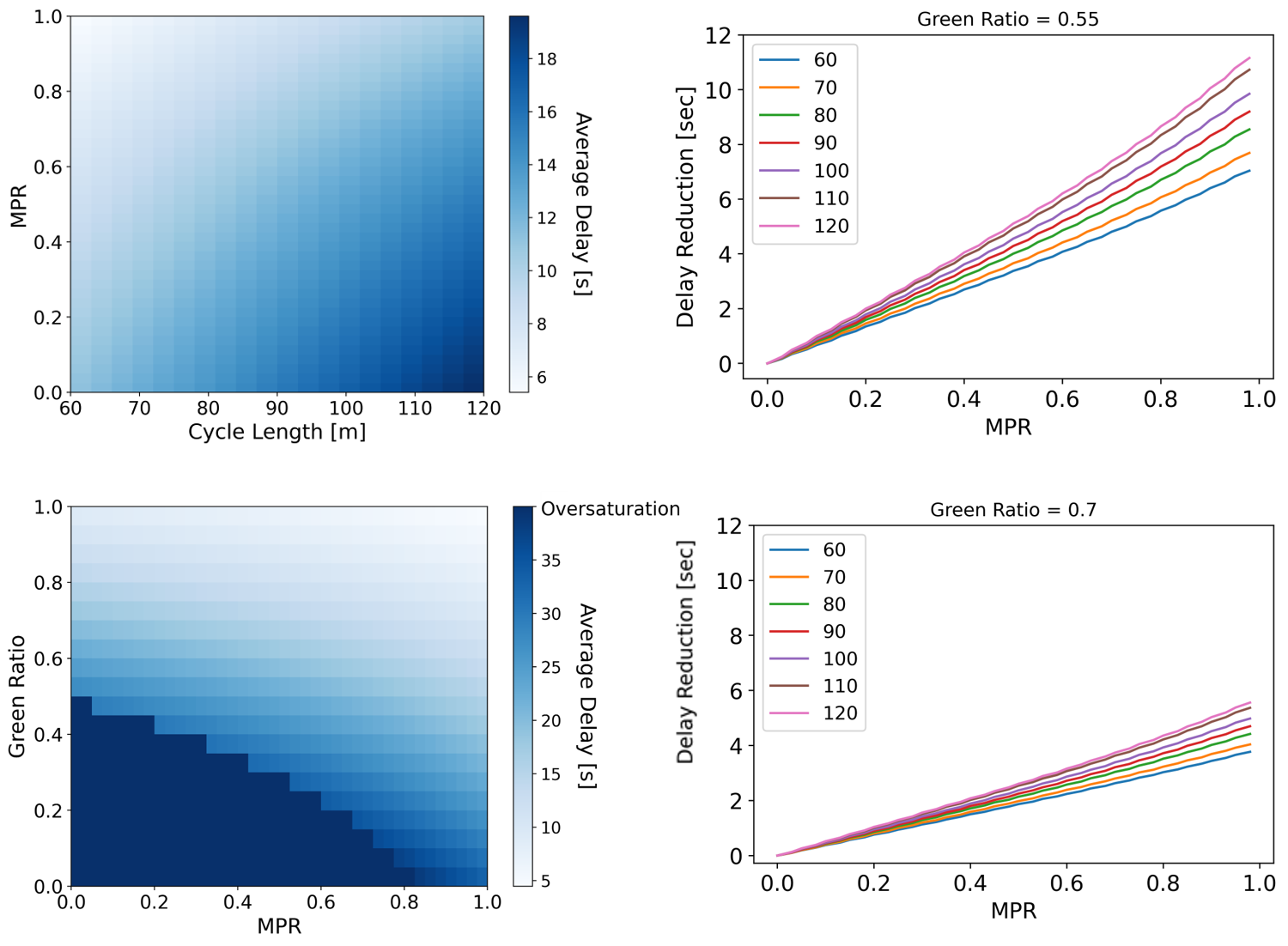}
    \caption{The joint impact on delay time of CAV ($p$ vs $c$) and ($G/c$ ratio vs $p$)}\label{fig:joint}
\end{figure}
The dark blue area indicates an over-saturation condition of the intersection, which is not accounted for in our model. Our results demonstrate that the expected average delay decreases as the market penetration rate $p$ and the green ratio $G/C$ increase. As the green ratio increases, the impact of increasing the market penetration rate on reducing delay will become less significant. Figure \ref{fig:reduction} further illustrates the changes in delay reduction rates across different market penetration rates and cycle lengths. 
\begin{figure}[h!]
\centering
    \includegraphics[width=0.47\textwidth]{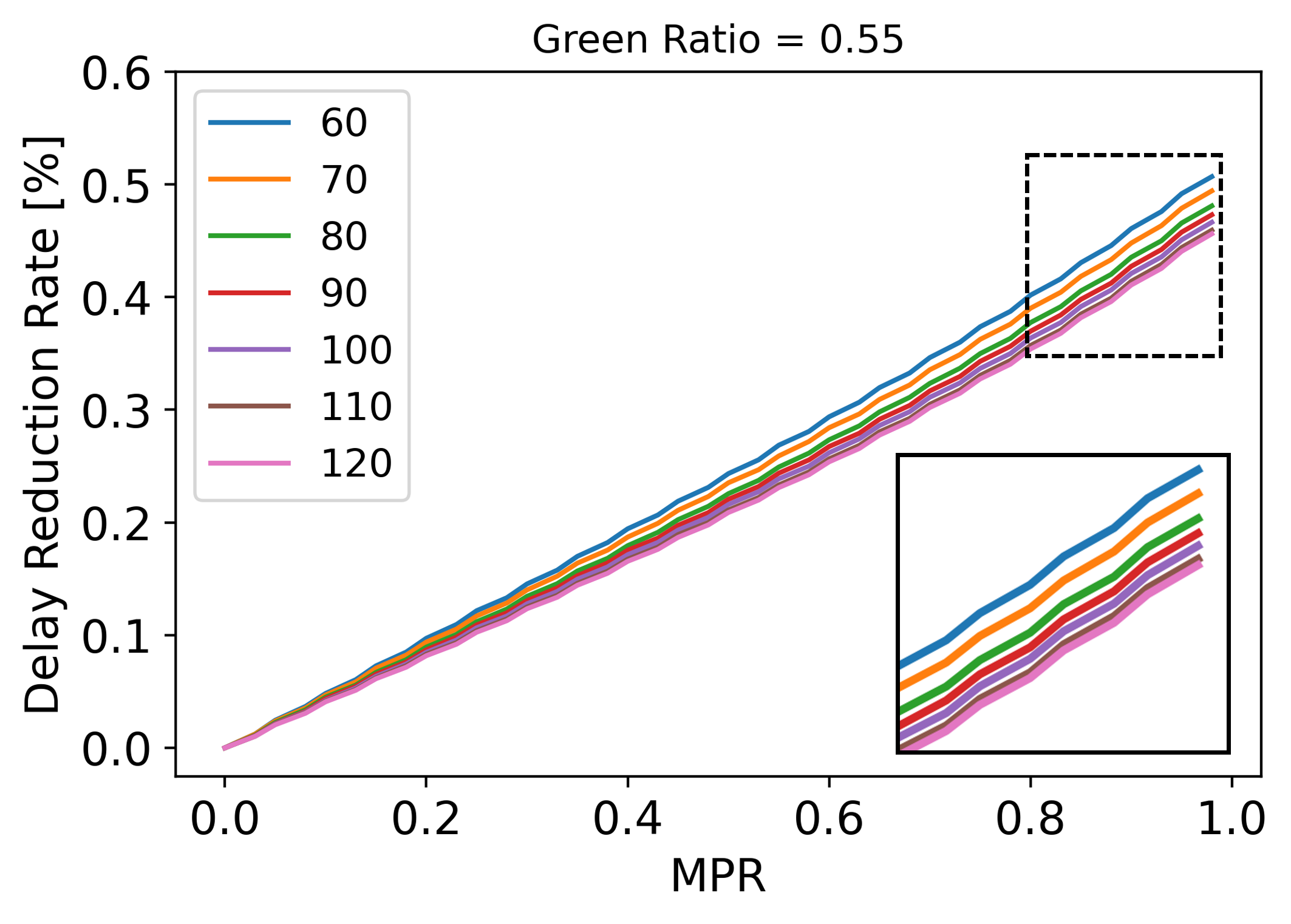}
    \caption{Delay Reduction Rate across different market penetration rate and cycle length}\label{fig:reduction}
\end{figure}
In summary, when the market penetration rate remains constant, an increase in cycle length leads to a lower delay reduction rate, which is consistent with the findings of \citep{koonce2008traffic}. Cycle length selection plays a critical role in intersection performance, as longer cycle length often result in increased delays due to extended red signal times and reduced responsiveness to fluctuating traffic demand.

Generally, numerical experiments have enhanced our understanding of the mixed traffic delay model in isolated intersections. Our numerical experiments  results reveal a negative correlation between the average delay of the intersection and both the green ratio and the market penetration rate. In contrast, the average delay positively correlates with the traffic arrival rate $q$ and cycle length $C$. In general, the results of our study align with the anticipated benefits of CAV in alleviating intersection congestion. Beyond these intuitive trends, a more detailed analysis reveals that the green ratio has a more significant influence on delay reduction compared to the market penetration rate of CAVs, particularly at moderate to high arrival rates. Furthermore, the sensitivity of delay to changes in the cycle length becomes more pronounced when the green ratio is relatively low, suggesting that optimizing the green time allocation is more effective than simply increasing the cycle length. Overall, these findings provide refined insights into intersection control strategies in mixed traffic environments, highlighting the critical importance of maintaining an efficient green ratio even as CAV adoption increases.

\subsection{Optimal cycle length for mixed traffic}
To gain a comprehensive understanding of the optimal cycle length $C^{*}$ for mixed traffic at isolated intersections, we performed numerical experiments on various CAV market penetration rates $p$ (from 0 to 1) to investigate the impact of increasing penetration rates of CAVs on traffic signal optimization and average delay reduction. The penetration rate of the CAVs is set from 0.01 to 0.99 with an increment of 0.02. As illustrated in Figure \ref{fig:delay}, the optimal cycle length decreases with increasing CAV market penetration rates. At a penetration rate of 0.01, the optimal cycle length is approximately 220 seconds. This length gradually reduces to 15 seconds as the penetration rate gradually approaches 1. Similarly, the average delay at the intersection reduces as the CAV market penetration rate increases. Initial penetration rates (0.01) see average delays around 60 seconds. This delay decrease progressively, nearing 4.5 seconds per cycle at nearly pure CAVs. The analysis demonstrates a clear correlation between higher penetration rates and enhanced traffic efficiency because the departure behavior of a CAV-led platoon can pass the intersection without start-up delay. In addition, a shorter cycle length results in a shorter average delay.

\begin{figure}[ht!]
\centering
    \includegraphics[width=0.47\textwidth]{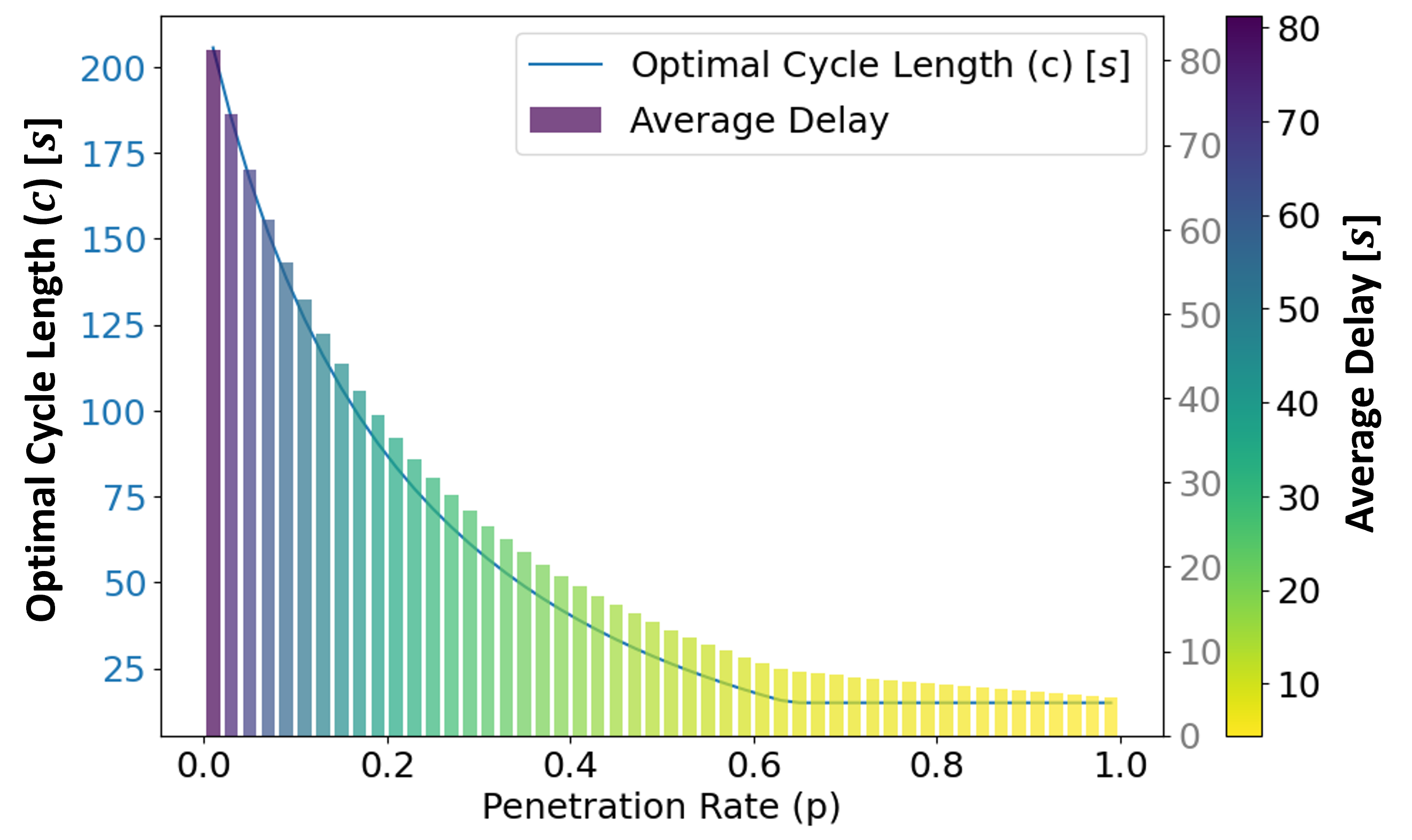}
    \caption{Penetration Rate vs. Optimal Cycle Length and Average Delay}\label{fig:delay}
\end{figure}

\section{CONCLUSION}
This paper presents a Markov-based analytical framework for approximating mixed traffic delay at an isolated signalized intersection and derives the optimal cycle length based on the proposed delay model. The framework systematically incorporates key factors, including CAV market penetration rate, traffic arrival and departure behaviors of intersections, communication capability, and CAV control parameters. Building upon the stochastic capacity analysis of pure CAV traffic developed by \cite{chen2023stochastic}, we extend the formulation to derive the expected stochastic capacity for mixed traffic as a function of the CAV penetration rate. Based on this capacity analysis, we develop an expected delay model for isolated signalized intersections that accounts for the distinct arrival and departure behaviors of CAV-led and HDV-led platoons.

The numerical experiments align with existing studies and further reveal three key insights. First, higher CAV market penetration consistently reduces average delay, driven by improved signal coordination and shorter headways. Second, signal timing optimization remains a critical factor in delay reduction and can yield greater benefits than increasing CAV penetration alone, especially under moderate to high traffic demand. Third, the optimal cycle length decreases as CAV penetration increases, due to shorter start-up delays and smoother trajectories that enable faster and more efficient intersection clearance. The proposed mixed traffic delay model for isolated signalized intersections can be feasibly integrated into existing signal control systems to support planning and evaluation in emerging mixed traffic scenarios.

Some future work can be implemented based on the current results. For instance, communication loss of CAVs can be modeled as a stochastic process in our framework, resulting in an additional departure behavior of the platoon, such as an AV leader, and a more complex transition matrix for the discrete Markov chain. Second, a more realistic traffic environment could be considered, such as stochastic arrival flows. Furthermore, the proposed analytical framework can be extended to over-saturated conditions to ensure generality. Additionally, the model could be validated and calibrated using simulation or field-test datasets to enhance its applicability.

\bibliographystyle{TRR} 
\bibliography{TRR_LaTeX_Guidelines}





\begin{ac}
The authors confirm contribution to the paper as follows: 
study conception and design: X. Yue, Y. Zhou, Z. Li, Y. Zhang; 
analysis and interpretation of results: X. Yue, Z. Li, Y. Zhou; 
draft manuscript preparation: X. Yue, Z. Li, Y. Zhou. 
All authors reviewed the results and approved the final version of the manuscript.
\end{ac}

\end{document}